\documentclass{snapshotmfo}
\usepackage[utf8]{inputenc}
\usepackage[dvipsnames]{xcolor}
\usepackage{mathtools}
\usepackage{enumitem}
\usepackage{amsmath,amssymb}
\usepackage{amsthm}
\usepackage{tcolorbox}
\usepackage{epigraph}
\usepackage{hyperref}
\hypersetup{
    colorlinks=true,
    linkcolor=Aquamarine,
    filecolor=Aquamarine,      
    urlcolor=Aquamarine,
    citecolor=Aquamarine,
}

\usepackage[round]{natbib}
\usepackage[USenglish]{babel}

\usepackage{ellipsis}

\newtheorem{theorem}{Theorem}[section]
\newtheorem{proposition}[theorem]{Proposition}
\newtheorem{lemma}[theorem]{Lemma}
\newtheorem{remark}[theorem]{Remark}
\newtheorem{corollary}[theorem]{Corollary}
\newtheorem{claim}[theorem]{Claim}
\newtheorem{definition}[theorem]{Definition}

\author{Mark Whitmeyer \thanks{Arizona State University, \href{mailto:mark.whitmeyer@gmail.com}{mark.whitmeyer@gmail.com}. I thank Joseph Whitmeyer for his comments. Special thanks is due to Michelle Whitmeyer, who has already taught me a lot about patience.}}
\title{Comparative Patience}

\begin{document}

\begin{abstract}
We begin by formulating and characterizing a dominance criterion for prize sequences: \(x\) dominates \(y\) if any impatient agent prefers \(x\) to \(y\). With this in hand, we define a notion of comparative patience. Alice is more patient than Bob if Alice's normalized discounted utility gain by going from any \(y\) to any dominating \(x\) is less than Bob's discounted utility gain from such an improvement. We provide a full characterization of this relation in terms of the agents' discount rules.
\end{abstract}

\newpage

\section{Introduction}

Typical economic agents prefer that good things happen sooner rather than later. That is, they are impatient. This paper seeks to clarify what exactly it means for one person to be more patient than another. 

One possible definition of comparative patience is the na\"{i}ve one: Alice is more patient than Bob if for any stream of payoffs, Alice's discounted sum of utilities is larger than Bob's. This is the modal definition of comparative patience in the literature--see for instance the literature on repeated games, where ``sufficiently patient'' refers to an exponential discount factor that is close enough to one.

Is this really what it means to be more patient, however? In particular, suppose Alice's discounted sum of utilities is higher than Bob's for some stream of payoffs, and it remains higher after we delay the arrival rate of the awards, yet Alice suffers a loss in her value for the stream as a result of the delay whereas Bob is unaffected. From this perspective, it seems like Bob ought to be characterized as more patient than Alice, as he is unfazed by the delay.

Our agenda in this paper, therefore, is to provide a behavioral characterization of what it means for one agent to be more patient than another. Here is how we define greater patience: Alice is more patient than Bob if the normalized difference between the discounted sums of her payoffs for any two payoff sequences with one worse (more delayed) than the other is lower than Bob's. Of course, this definition relies heavily on what we mean by ``worse'' in the context of payoff sequences. Thus, our first task is to define and characterize this concept.

We say that one prize sequence \(x\) is superior to another sequence \(y\) if \(\sum_{t=1}^{T} \beta_t x_t \geq \sum_{t=1}^{T} \beta_t y_t\) for any positive decreasing sequence \(\beta = \left\{\beta_t\right\}_{t=1}^{T}\), where the time horizon \(T\) is either finite or infinite. In words, \(x\) is superior to \(y\) if any impatient agent (one who discounts the future) prefers \(x\) to \(y\). Inferiority in this sense, then, is how we formalize one sequence being worse than another. Our first result, Proposition \ref{superiority}, characterizes this superiority in terms of the sequences \(x\) and \(y\) themselves.

We then turn our attention to comparative patience, defined as follows. Alice, who discounts according to discounting sequence \(\alpha\), is more patient than Bob, who discounts according to \(\beta\), if for any positive prize sequences \(x\) and \(y\) with \(x\) superior to \(y\) and \(\sum_{t=1}^{T} x_t = \sum_{t=1}^{T} y_t\) (and such that the denominators of the following are nonzero),
\[\frac{\sum_{t=1}^{T} \alpha_t x_t}{\sum_{t=1}^{T} \alpha_t y_t} \leq \frac{\sum_{t=1}^{T} \beta_t x_t}{\sum_{t=1}^{T} \beta_t y_t}\text{.}\]
This is equivalent to the normalized difference in Alice's discounted payoffs from \(x\) and \(y\) being less than the difference of Bob's discounted payoffs from the two streams.

In the main result of the paper, Theorem \ref{maintheorem}, we show that Alice is more patient than Bob if and only if \[\frac{\beta_t - \beta_{t+1}}{\alpha_t - \alpha_{t+1}} \geq \frac{\beta_1}{\alpha_1}\] \textit{for every} \(t\). This is an exceptionally strong condition; for one, it implies that \(\alpha_t/\beta_t\) is monotonically increasing in \(t\). Notably, it is only if the time horizon is two that it is satisfied for any two exponential discounters. That is, if there are three or more periods Alice can have exponential discount factor \(a\) and Bob \(b\), with \(a > b\), yet Alice is not more patient than Bob.

Economists have long been interested in the time preferences of agents (\cite{von1891positive}, \cite{fisher1930theory}). Many--\cite{koopmans1960stationary}, \cite{burness1973impatience}, \cite{burness1973impatience}, and \cite{nachman1975risk}--have explored patience (impatience) and how agents evaluate consumption streams, or at least ought to. We ignore these important issues. Instead, we merely specify that agents' utility streams over per-period rewards are additively separable over time, and that the agents are impatient in the sense that they weigh these utilities according to a decreasing discount function.

Two works are especially close to this one. The first is the working-paper version \cite{chambers2018multiple}. There, the authors provide a robust ranking of reward streams, where one dominates another if it has a higher discounted sum for all exponential discounters. The first part of this paper, instead, provides a robust ranking of reward streams, where the superiority is required to hold for all discounters rather than exponential ones.

The second similar work is \cite{quah2013discounting}, who explore the effect of an agent's discount rate on her behavior in stopping problems. They formulate a ``more patient than'' partial order. In this conception, Alice is more patient than Bob if \(\alpha_t/\beta_t\) is increasing in \(t\). They show that this is equivalent to Bob preferring a later prize to an earlier one implying Alice does too. In general, the patience order in this paper is strictly stronger than \citeauthor{quah2013discounting}'s, though in Proposition \ref{twoserene}, we show that when there are just two periods the relations are equivalent. For three or more periods; however, they are not, and, in fact, a more patient Alice in \citeauthor{quah2013discounting}'s sense may suffer a greater normalized loss from moving from a superior to an inferior stream than Bob.


\section{Dominance of Sequences}

Throughout the paper, we make the standing assumption that all infinite sequences under consideration are absolutely summable. For \(T \in \mathbb{N} \cup \left\{\infty\right\}\), let \(x = \left\{x_t\right\}_{t=1}^{T}\) and \(y = \left\{y_t\right\}_{t=1}^{T}\) be two prize sequences (in utils). We define \[\mathcal{T} \coloneqq \left\{t \in \mathbb{N} \colon t \leq T\right\}\text{.}\] A sequence, \(\beta\), is \textcolor{Aquamarine}{positive} if \(\beta_t \geq 0\) for all \(t \in \mathcal{T}\) and \textcolor{Aquamarine}{decreasing} if \(\beta_t \geq \beta_{t+1}\) for all \(t \in \mathcal{T}\). Here are three notions of superiority of prize sequences.
\begin{definition}
    Prize sequence \(x\) \textcolor{Aquamarine}{dominates} \(y\) \textcolor{Aquamarine}{pointwise} if \(x_t \geq y_t\) for all \(t \in \mathcal{T}\).
\end{definition}

\begin{definition}
    Prize sequence \(x\) \textcolor{Aquamarine}{dominates} \(y\) if \(\sum_{t=1}^{p} x_t \geq \sum_{t=1}^{p} y_t\) for all \(p \in \mathcal{T}\). We write this \(x \trianglerighteq y\).
\end{definition}
Naturally, if \(x\) dominates \(y\) pointwise, \(x\) dominates \(y\), while the converse is false.
\begin{definition}
    Prize sequence \(x\) is \textcolor{Aquamarine}{superior} to \(y\) if \(\sum_{t=1}^{T} \beta_t x_t \geq \sum_{t=1}^{T} \beta_t y_t\) for any positive decreasing sequence \(\beta = \left\{\beta_t\right\}_{t=1}^{T}\).
\end{definition}

We note the following lemma. Its proof, as well as all proofs omitted from the main text, may be found in the Appendix (\ref{appendix}).
\begin{lemma}\label{pointlemma}
    If \(x \trianglerighteq y\), there exists a \(\tilde{x}\) that \(x\) dominates pointwise, and that satisfies \(\sum_{t=1}^{T} \tilde{x}_t = \sum_{t=1}^{T} y_t\) and
    \(\sum_{t=1}^{p} \tilde{x}_t \geq \sum_{t=1}^{p} y_t\) for all \(p \in \mathcal{T}\).
\end{lemma}
This lemma is unsurprising. If the sum of the elements of \(x\) is strictly higher than the sum of the elements of \(y\), we can just remove excess amounts from particular elements of \(x\) in order to obtain a balance (equal sums), while preserving the dominance of the partial sums. Next we state the following lemma, which is a discrete version of the integration by parts formula.
\begin{lemma}[Abel's Lemma]\label{abel}
    Let \(a\) and \(b\) be two sequences of real numbers. For all \(k \in \mathcal{T}\), we define \(A_k = \sum_{i=1}^k a_i\). Then, for \(T \in \mathbb{N}\),
    \[\sum_{t=1}^T a_t b_t = \sum_{t=1}^{T-1} A_t (b_t - b_{t+1}) + A_T b_t\text{,}\]
    and for \(T = \infty\),
    \[\sum_{t=1}^\infty a_t b_t = \sum_{t=1}^\infty A_t (b_t - b_{t+1})\text{.}\]
\end{lemma}
Using Abel's lemma, we obtain our next result. This is almost exactly the characterization of first-order stochastic dominance--once we normalize the sums so that they sum to one--with the single exception that elements of \(\tilde{x}\) and \(y\) may take negative values.
\begin{lemma}\label{abel2}
    Let \(\tilde{x}\) and \(y\) be two sequences that satisfy \(\sum_{t=1}^{T} \tilde{x}_t = \sum_{t=1}^{T} y_t\) and
    \(\sum_{t=1}^{p} \tilde{x}_t \geq \sum_{t=1}^{p} y_t\) for all \(p \in \mathcal{T}\). Then, for any positive decreasing sequence \(\beta\), \(\sum_{t=1}^{T} \beta_t \tilde{x}_t \geq \sum_{t=1}^{T}\beta_t y_t\).
\end{lemma}

With these lemmas in hand, we equate superiority with dominance.
\begin{proposition}\label{superiority}
    \(x\) is superior to \(y\) if and only if \(x\) dominates \(y\).
\end{proposition}
\begin{proof}
    \(\left(\Leftarrow\right)\) Let \(x\) dominate \(y\). Then, for any positive decreasing sequence \(\beta\),
    \[\sum_{t=1}^{T} \beta_t x_t \geq \sum_{t=1}^{T} \beta_t \tilde{x}_t \geq \sum_{t=1}^{T}\beta_t y_t\text{,}\]
    where the first inequality follows from Lemma \ref{pointlemma} and the second from Lemma \ref{abel2}.

    \medskip

    \noindent \(\left(\Rightarrow\right)\) Suppose for the sake of contraposition that \(x\) does not dominate \(y\). Then, there exists some \(p \in \mathbb{N}\) such that 
    \[\sum_{t=1}^{p} x_t < \sum_{t=1}^{p} y_t\text{.}\]
    Then, let \(\beta_t =1\) for all \(t \in \left\{1,\dots,p\right\}\) and \(\beta_t = 0\) for all \(t \in \left\{p+1,\dots,T\right\}\). Consequently,
    \[\sum_{t=1}^{T} \beta_t x_t = \sum_{t=1}^p x_t < \sum_{t=1}^{p} y_t = \sum_{t=1}^{T} \beta_t y_t\text{.}\]
    and so \(x\) is not superior to \(y\).
\end{proof}
This proposition is more-or-less first-order stochastic dominance, which makes sense due to the mononicity of \(\beta\).

\section{Relative Patience}

Now that we have provided a ranking of prize streams, we can explore relative patience. Henceforth, we specify that any discounting sequence, \(\beta\), besides being absolutely summable and decreasing, is such that \(\beta_t > 0\) for any \(t \in \mathcal{T}\), i.e., is strictly positive.

A first pass at defining what it means for one agent to be more patient than another is the following.
\begin{definition}
    Alice, who discounts according to \(\alpha\), is more \textcolor{Aquamarine}{serene} than Bob, who discounts according to \(\beta\), if for any prize sequences \(x\) and \(y\) with \(x \trianglerighteq y\) and \(\sum_{t=1}^{T} x_t = \sum_{t=1}^{T} y_t\),
    \[\sum_{t=1}^{T} \alpha_t x_t - \sum_{t=1}^{T} \alpha_t y_t \leq \sum_{t=1}^{T} \beta_t x_t - \sum_{t=1}^{T} \beta_t y_t \text{.}\]
\end{definition}
We can characterize relative serenity in terms of the agents' discounting behavior.
\begin{proposition}\label{moreserene}
    Alice is more serene than Bob if and only if \(\beta - \alpha\) is a decreasing sequence with \(\beta_t \geq \alpha_t\) for all \(t \in \mathcal{T}\).
\end{proposition}
\begin{proof}
Obviously, for \(x\) and \(y\),
\[\sum_{t=1}^{T} \alpha_t x_t - \sum_{t=1}^{T} \alpha_t y_t \leq \sum_{t=1}^{T} \beta_t x_t - \sum_{t=1}^{T} \beta_t y_t \ \Leftrightarrow \ \sum_{t=1}^{T} \left(\beta_t - \alpha_t\right) y_t \leq \sum_{t=1}^{T} \left(\beta_t - \alpha_t\right) x_t \text{.}\]

\medskip

\noindent \(\left(\Leftarrow\right)\) Suppose \(\beta - \alpha\) is a decreasing sequence with \(\beta_t \geq \alpha_t\) for all \(t \in \mathcal{T}\). Then, as \(x \trianglerighteq y\), \[\sum_{t=1}^{T} \left(\beta_t - \alpha_t\right) x_t \geq \sum_{t=1}^{T} \left(\beta_t - \alpha_t\right) y_t\text{.}\]

\medskip

\noindent \(\left(\Rightarrow\right)\) Suppose first for the sake of contraposition that \(\beta - \alpha\) is not a decreasing sequence. Then there exist \(k, k+1 \in \mathcal{T}\) such that \(\beta_k - \alpha_k < \beta_{k+1} - \alpha_{k+1}\). Define \(x_t = 0\) for all \(t \neq k\) and \(x_{k} = 1\), and define \(y_t = 0\) for all \(t \neq k+1\) and \(y_{k+1} = 1\). By construction \(x \trianglerighteq y\), but
\[\sum_{t=1}^{T} \left(\beta_t - \alpha_t\right) x_t = \beta_k - \alpha_k  < \beta_{k+1} - \alpha_{k+1} = \sum_{t=1}^{T} \left(\beta_t - \alpha_t\right) y_t\text{,}\]
and so Alice is not more serene than Bob. 

Finally, suppose for the sake of contraposition that \(\beta - \alpha\) is a decreasing sequence, but that there exists \(k \in \mathcal{T}\) for which \(\beta_k - \alpha_k < 0\). But then let \(x_t = 0\) for all \(t \in \mathcal{T}\), \(y_k = -1\) and \(y_t = 0\) for all \(t \in \mathcal{T} \setminus \left\{k\right\}\). So, 
\[\sum_{t=1}^{T} \left(\beta_t - \alpha_t\right) x_t = 0  < \alpha_k - \beta_k = \sum_{t=1}^{T} \left(\beta_t - \alpha_t\right) y_t\text{,}\]
and so Alice is not more serene than Bob.
\end{proof}
When Alice and Bob are both exponential discounters, i.e., \(\alpha_t = a^{t-1}\) and \(\beta_t = b^{t-1}\), \(\beta_2 \geq \alpha_2\) if and only if \(b \geq a\). However, \(\beta_1 - \alpha_1 = 0\), so \(\beta - \alpha\) is decreasing only if \(a \geq b\). Accordingly, 
\begin{remark}
    If Alice and Bob are exponential discounters, Alice is more serene than Bob if and only if their discount factors are the same.
\end{remark}
More generally,
\begin{corollary}\label{expserene}
    Let \(\alpha_1 = \beta_1\). Then, Alice is more serene than Bob if and only if \(\alpha = \beta\).
\end{corollary}
This corollary is easy: we need \(\beta - \alpha\) to be decreasing, but the difference to also be weakly positive. As there is nowhere to go but down and the difference is starting at \(0\), it must remain at \(0\). Proposition \ref{moreserene} and Corollary \ref{expserene} reveal that this definition is not terribly satisfactory. Can we improve upon it?

Upon reflecting the definition of ``more serene than,'' we see that it could be viewed as a bit of an ``apples to oranges'' comparison: \(\sum_{t=1}^{T} \alpha_t y_t\) and  \(\sum_{t=1}^{T} \beta_t y_t\) need not be the same, so perhaps the proper definition should normalize things. In any case, let us take a second pass at defining comparative patience.

If both \(x\) and \(y\) are positive sequences, \(\sum_{t=1}^{T} \alpha_t y_t \neq 0\), and \(\sum_{t=1}^{T} \beta_t y_t \neq 0\), we can normalize Alice's payoff difference by \[\frac{\sum_{t=1}^{T} \beta_t y_t}{\sum_{t=1}^{T} \alpha_t y_t}\text{.}\]
Doing so, we note that 
\[\frac{\sum_{t=1}^{T} \beta_t y_t}{\sum_{t=1}^{T} \alpha_t y_t}\left(\sum_{t=1}^{T} \alpha_t x_t - \sum_{t=1}^{T} \alpha_t y_t\right) \leq \sum_{t=1}^{T} \beta_t x_t - \sum_{t=1}^{T} \beta_t y_t \quad \Leftrightarrow \quad \frac{\sum_{t=1}^{T} \alpha_t x_t}{\sum_{t=1}^{T} \alpha_t y_t} \leq \frac{\sum_{t=1}^{T} \beta_t x_t}{\sum_{t=1}^{T} \beta_t y_t}\text{.}\]

\begin{definition}
    Alice, who discounts according to \(\alpha\), is more \textcolor{Aquamarine}{patient} than Bob, who discounts according to \(\beta\), if for any positive prize sequences \(x\) and \(y\) with \(x \trianglerighteq y\) and \(\sum_{t=1}^{T} x_t = \sum_{t=1}^{T} y_t\), such that \(\sum_{t=1}^{T} \alpha_t y_t \neq 0\) and \(\sum_{t=1}^{T} \beta_t y_t \neq 0\),
\[\frac{\sum_{t=1}^{T} \alpha_t x_t}{\sum_{t=1}^{T} \alpha_t y_t} \leq \frac{\sum_{t=1}^{T} \beta_t x_t}{\sum_{t=1}^{T} \beta_t y_t}\text{.}\]
\end{definition}
Of course, as \(x \trianglerighteq y\), \(\sum_{t=1}^{T} \alpha_t y_t \neq 0\) and \(\sum_{t=1}^{T} \beta_t y_t \neq 0\) imply \(\sum_{t=1}^{T} \alpha_t x_t \neq 0\) and \(\sum_{t=1}^{T} \beta_t x_t \neq 0\). Furthermore, \[\frac{\sum_{t=1}^{\infty} \alpha_t x_t}{\sum_{t=1}^{\infty} \alpha_t y_t} \leq \frac{\sum_{t=1}^{\infty} \beta_t x_t}{\sum_{t=1}^{\infty} \beta_t y_t} \quad \Leftrightarrow \quad \frac{\sum_{t=1}^{\infty} \beta_t y_t}{\sum_{t=1}^{\infty} \alpha_t y_t} \leq \frac{\sum_{t=1}^{\infty} \beta_t x_t}{\sum_{t=1}^{\infty} \alpha_t x_t}\text{.}\]

When \(T = 2\), there is a nice characterization of patience in terms of the discount factors.
\begin{proposition}\label{twoserene}
    If \(T = 2\), Alice is more patient than Bob if and only if \(\frac{\alpha_1}{\beta_1} \leq \frac{\alpha_2}{\beta_2}\).
\end{proposition}
As we noted in the introduction, this is precisely \cite{quah2013discounting}'s definition of comparative patience. This makes sense, their definition of comparative patience--that if Bob prefers a later prize to an earlier one, so too does a more patient Alice--inherently (as the discounting is free) makes theirs a statement about two-period problems.

    Alas, it is not the case that \[\frac{\alpha_1}{\beta_1} \leq \frac{\alpha_2}{\beta_2} \leq \frac{\alpha_3}{\beta_3}\] implies Alice is more patient than Bob when \(T = 3\). To see this, let 
    \[\alpha_1 = \frac{1}{2}, \ \alpha_2 = \frac{12}{25}, \ \alpha_3 = \frac{91}{250}, \ \beta_1 = 1, \ \beta_2 = \frac{2}{3}, \text{ and } \beta_3 = \frac{1}{2}\text{.}\]
    We have
    \[\frac{1}{2} < \frac{18}{25} < \frac{91}{125}\text{,}\] so  \(\alpha/\beta\) is monotone. However, if
    \[x_1 = y_1 = 1, \ x_2 = y_3 = \frac{3}{2}, \text{ and } x_3 = y_2 = 1\text{,}\]
    \[\frac{\alpha_1 x_1 + \alpha_2 x_2 + \alpha_3 x_3}{\alpha_1 y_1 + \alpha_2 y_2 + \alpha_3 y_3} - \frac{\beta_1 x_1 + \beta_2 x_2 + \beta_3 x_3}{\beta_1 y_1 + \beta_2 y_2 + \beta_3 y_3} > .0035 > 0 \text{,}\]
    so Alice is not more patient than Bob.

    The following observation will be useful. For two positive sequences \(x\) and \(z\), \(z\) is a \textcolor{Aquamarine}{binary deterioration} of \(x\) if there exist \(t_1, t_2 \in \mathcal{T}\) (with \(t_1 < t_2\)), and \(\eta > 0\) such that
    \(x_t = z_t\) for all \(t \neq t_1\) and \(t \neq t_2\), \(z_{t_1} = x_{t_1} - \eta\), and \(z_{t_2} = x_{t_2} + \eta\).
\begin{lemma}\label{deteriorate}
    Let \(\sum_{t=1}^{T} x_t = \sum_{t=1}^{T} y_t\). \(x \trianglerighteq y\) if and only if there exists a sequence of binary deteriorations \(x, z^1, \dots, z^k, y\).
\end{lemma}
\begin{proof}
    \(\left(\Leftarrow\right)\) This direction is immediate: by construction,
    \[x \trianglerighteq z^1 \trianglerighteq \dots \trianglerighteq z^k \trianglerighteq y\text{.}\]

    \medskip

    \noindent \(\left(\Rightarrow\right)\) If \(x\) and \(y\) differ at only finitely many values of \(t\), the result follows from the discussion on pages 630 and 631 of \cite{fishburn1982moment}. It is extremely likely that this result is known in the infinite case, but we prove it anyway, in Appendix \ref{deteriorateproof} \end{proof}



We define 
\[\bar{\mathcal{T}} \coloneqq \begin{cases}
     t \in \mathbb{N} \colon t \leq T-1, \quad &\text{if} \quad T \in \mathbb{N}\\
     \mathbb{N}, \quad &\text{if} \quad T = \infty\text{.}
\end{cases}\] 
Here is the main characterization result. 
\begin{theorem}\label{maintheorem}
    Alice is more patient than Bob if and only if
    \[\inf_{t \in \bar{\mathcal{T}}}\frac{\beta_t - \beta_{t+1}}{\alpha_t - \alpha_{t+1}} \geq \frac{\beta_1}{\alpha_1}\text{.}\]
\end{theorem}
This condition is strictly stronger than monotonicity of \(\alpha_t/\beta_t\), and, as the three-period example above reveals, is equivalent to such monotonicity if and only if there are two periods.
\begin{corollary}
    If Alice is more patient than Bob, \(\alpha_t/\beta_t\) is increasing in \(t\).
\end{corollary}
It turns out that if Alice is more patient than Bob, and Bob is more patient than Michelle, Alice is more patient than Michelle. \textit{Viz.,}
    \begin{proposition}\label{transpatient}
        Patience is transitive.
    \end{proposition}

    The two most common discounting sequences are exponential discounting and quasi-hyperbolic discounting (\cite{laibson1997golden}). Observe that if Alice's and Bob's discounting rules are both in one of these classes, \(\alpha_1 = \beta_1 = 1\). That inspires this corollary:
    \begin{corollary}\label{corollary311}
        Let \(\alpha_1 = \beta_1\). Then, Alice is more patient than Bob if and only if \(\beta_t - \beta_{t+1} \geq \alpha_t - \alpha_{t+1}\) for all \(t \in \bar{\mathcal{T}}\).
    \end{corollary}
    Interestingly, we see that for \(T \geq 3\), it could be that Alice and Bob are both exponential discounters, with \(\alpha_t = a^{t-1}\), \(\beta_t = b^{t-1}\), and \(1 > a > b > 0\), yet Alice is not more patient than Bob. A similar incomparibility is true when both are quasi-hyperbolic discounters. Nevertheless, we observe
    \begin{corollary}
        Let Alice and Bob be exponential discounters and \(T \in \mathbb{N}\). For any discounting rule for Bob, \(\beta_t = b^{t-1}\) (\(b \in \left(0,1\right)\)), there exists an \(\eta > 0\) such that if \(\alpha_t = a^{t-1}\) for \(a \in \left(1-\eta,1\right)\), Alice is more patient than Bob.
    \end{corollary}
    \begin{proof}
        \(b^{t-1} - b^{t}\) is strictly decreasing in \(t\) and \(\beta^{T-1} - \beta^T > 0\). Moreover \(a^{t-1} - a^{t}\) is continuous in \(a\) for all \(t \in \bar{\mathcal{T}}\) and equals \(0\) for all \(t \in \bar{\mathcal{T}}\) when \(a = 1\). Thus, the result follows from the intermediate-value theorem and Corollary \ref{corollary311}.    \end{proof}
        On the other hand, the infinite-horizon environment is troubling.
        \begin{corollary}
        Let Alice and Bob be exponential discounters and \(T = \infty\). Alice is more patient than Bob if and only if they have the same discount factor.
    \end{corollary}
    \begin{proof}
        Let the discount factor for Alice \(a \in \left(0,1\right)\) be strictly greater than that for Bob, \(b\). Observe that for any \(t \in \mathbb{N}\),
        \[\frac{\beta_t - \beta_{t+1}}{\alpha_t - \alpha_{t+1}} = \left(\frac{b}{a}\right)^{t-1}\frac{1-b}{1-a}\text{.}\]
        This expression is strictly decreasing in \(t\) and equals \(0\) when \(t \to \infty\). Thus, there exists a \(t^{\dagger} \in \mathbb{N}\) such that this expression is strictly less than \(1\) for all \(t \geq t^{\dagger}\). We conclude that Alice is not more patient than Bob.
    \end{proof}
    Finally, as the patience relation is antisymmetric when \(\beta_1 = \alpha_1\),
    \begin{corollary}
        If \(\beta_1 = \alpha_1\), patience is a partial order on a set of agents.
    \end{corollary}


\bibliography{sample.bib}

\appendix

\section{Omitted Proofs}\label{appendix}

\subsection{Lemma \ref{pointlemma} Proof}
\begin{proof}
First, suppose \(T = \infty\). The sequence \(\left\{\sum_{i=1}^t \left(x_i-y_i\right)\right\}_{t=1}^{\infty}\) converges and is, therefore, bounded. Consequently, for any \(t \in \mathbb{N}\),  \[s_t \coloneqq \inf_{k \geq t} \sum_{i=1}^k \left(x_i-y_i\right)\] is well-defined. Then, for any \(t \in \mathbb{N}\), we define \(z_t = s_{t} - s_{t-1}\), where \(s_0 = 0\). By construction \(z_t \geq 0\) for all \(t \in \mathbb{N}\). We define \(\tilde{x}_t = x_t - z_t\) for all \(t \in \mathbb{N}\) and observe that, by construction,
\[\sum_{t=1}^{p} \tilde{x}_t = \sum_{t=1}^{p} \left(x_t - z_t\right) = \sum_{t=1}^{p} x_t - s_p = \sum_{t=1}^{p} x_t - \inf_{k \geq p} \sum_{i=1}^k \left(x_i-y_i\right) \geq \sum_{t=1}^{p} y_t\text{,}\] for all \(p \in \mathbb{N}\). Moreover,
\[\sum_{t=1}^{\infty} \tilde{x}_t = \lim_{p \to \infty} \sum_{t=1}^{p} \tilde{x}_t = \sum_{t=1}^{\infty} x_t - \liminf_{p \to \infty} \sum_{i=1}^p \left(x_i-y_i\right) = \sum_{t=1}^{\infty} y_t\text{,}\] as desired. 

If \(T \in \mathbb{N}\) the proof is virtually identical, \textit{mutatis mutandis}.\end{proof}

\subsection{Lemma \ref{abel2} Proof}
\begin{proof}
    Suppose \(T = \infty\). By Lemma \ref{abel},
    \[\sum_{t=1}^{\infty} \beta_t \tilde{x}_t - \sum_{t=1}^{\infty}\beta_t y_t = \sum_{t=1}^{\infty} \underbrace{\left(\beta_t - \beta_{t+1}\right)}_{\geq 0}\underbrace{\sum_{i=1}^{t} (\tilde{x}_i-y_i) }_{\geq 0} \geq 0\text{.}\]Now suppose \(T \in \mathbb{N}\). By Lemma \ref{abel},
    \[\sum_{t=1}^{T} \beta_t \tilde{x}_t - \sum_{t=1}^{T}\beta_t y_t = \sum_{t=1}^{T-1} \underbrace{\left(\beta_t - \beta_{t+1}\right)}_{\geq 0}\underbrace{\sum_{i=1}^{t} (\tilde{x}_i-y_i) }_{\geq 0} + \beta_T \underbrace{\sum_{i=1}^T \left(\tilde{x}_i - y_i\right)}_{ = 0} \geq 0\text{,}\]
    where the last term equals \(0\) by assumption.\end{proof}

    \subsection{Proposition \ref{twoserene} Proof}

\begin{proof}
    \(\left(\Leftarrow\right)\) Suppose \(\frac{\alpha_1}{\beta_1} \leq \frac{\alpha_2}{\beta_2}\). We have
    \[\frac{\beta_1 x_1 + \beta_2 x_2}{\beta_1 y_1 + \beta_2 y_2} = \frac{\beta_1 x_1 + \beta_2 x_2}{\beta_1 y_1 + \beta_2 (x_2+x_1-y_1)}\text{.}\]
    The derivative of this with respect to \(\beta_1\) simplifies to
    \[\frac{\beta_2\left(x_1+x_2\right)\left(x_1-y_1\right)}{\left(\beta_1 y_1 + \beta_2 (x_2+x_1-y_1)\right)^2} \geq 0\text{.}\]
    Consequently,
    \[\frac{\beta_1 x_1 + \beta_2 x_2}{\beta_1 y_1 + \beta_2 y_2} \geq \frac{x_1 \frac{\beta_2}{\alpha_2}\alpha_1 + \beta_2 x_2}{ y_1 \frac{\beta_2}{\alpha_2}\alpha_1 + \beta_2 y_2} = \frac{\alpha_1 x_1 + \alpha_2 x_2}{\alpha_1 y_1 + \alpha_2 y_2}\text{.}\]

    \medskip

    \noindent \(\left(\Rightarrow\right)\) Suppose for the sake of contraposition \(\frac{\alpha_1}{\beta_1} > \frac{\alpha_2}{\beta_2}\). Let \(x_1 = y_2 = 1\) and \(x_2 = y_1 = 0\). Thus,
    \[\frac{\alpha_1 x_1 + \alpha_2 x_2}{\alpha_1 y_1 + \alpha_2 y_2} = \frac{\alpha_1}{\alpha_2} > \frac{\beta_1}{\beta_2} = \frac{\beta_1 x_1 + \beta_2 x_2}{\beta_1 y_1 + \beta_2 y_2}\text{,}\]
    so Alice is not more patient than Bob.\end{proof}

    \subsection{Lemma \ref{deteriorate} Proof}\label{deteriorateproof}
    \begin{proof}
        Suppose \(x \trianglerighteq y\) and \(x \neq y\) (if \(x = y\), we are done). Let
    \[t_1 \coloneqq \min\left\{t \in \mathcal{T} \colon x_t \neq y_t\right\}\text{.}\]
    As \(x \trianglerighteq y\), \(x_{t_1} > y_{t_1}\). Now let
    \[t_{1}' \coloneqq \min\left\{t \in \mathcal{T} \colon x_t < y_t\right\}\text{.}\]
    As \(x \trianglerighteq y\), \(t_{1}' > t_1\). There are two possibilities: either \(x_{t_1} - y_{t_1} \geq y_{t_{1}'} - x_{t_{1}'}\), or \(x_{t_1} - y_{t_1} < y_{t_{1}'} - x_{t_{1}'}\). In either case, we define \(z^1_t = x_t\) for all \(t \neq t_1, t_{1}'\), \(z^1_{t_1} = x_{t_1} - \eta\) and \(z^1_{t_{1}'} = x_{t_{1}'} + \eta\); where \(\eta = y_{t_{1}'} - x_{t_{1}'}\) in the first case and \(\eta = x_{t_{1}} - y_{t_{1}}\) in the second.

    By construction \(z^1\) is a binary deterioration of \(x\). We then construct \(z^2\) from \(z^1\)--and, subsequently, each \(z^{i+1}\) from \(z^i\) (\(k-1 \geq i \geq 1\))--in the same manner. If \(T \in \mathbb{N}\), we may conclude the result. Let \(T = \infty\). Suppose for the sake of contradiction that there exists \(\eta > 0\) and a \(t \in \mathbb{N}\) such that \(z_t^j - y_t > \eta\) for all \(j \in \mathbb{N}\). By the construction of the sequence of \(z\)s, there exists \(k \in \mathbb{N}\) such that \(z_{i}^k = y_i\) for all \(i < k\). Moreover, by the construction of the sequence of \(z\)s, it must be that \(z_t^k - y_t - \sum_{i = t+1}^{\infty}\left(y_i - z_i^k\right) > \eta > 0\). However, this contradicts the fact that \(\sum_{t=1}^{\infty}\left(z^k_t-y_t\right) = 0\). 
    \end{proof}

    \subsection{Theorem \ref{maintheorem} Proof}
    \begin{proof}
    \(\left(\Rightarrow\right)\) Suppose for the sake of contraposition that \[\inf_{t \in \bar{\mathcal{T}}}\frac{\beta_t - \beta_{t+1}}{\alpha_t - \alpha_{t+1}} < \frac{\beta_1}{\alpha_1}\text{.}\] This implies that there exists a \(k \in \bar{\mathcal{T}}\) such that \[\frac{\beta_k - \beta_{k+1}}{\alpha_k - \alpha_{k+1}} < \frac{\beta_1}{\alpha_1}\text{.}\] 
    
    If \(k = 1\), this implies \(\alpha_1 \beta_2 > \alpha_2 \beta_1\), in which case, following the proof of Proposition \ref{twoserene}, we are done. Consequently, without loss of generality we specify \(k = 2\), i.e.,
    \[\frac{\beta_2 - \beta_{3}}{\alpha_2 - \alpha_{3}} < \frac{\beta_1}{\alpha_1}\text{.}\] 
    Moreover, if 
    \[\frac{\beta_2 - \beta_{3}}{\alpha_2 - \alpha_{3}} < \frac{\beta_2}{\alpha_2}\text{,}\] again following the proof of Proposition \ref{twoserene} we are done. Accordingly, we assume
    \[\frac{\beta_2}{\alpha_2} \leq \frac{\beta_2 - \beta_{3}}{\alpha_2 - \alpha_{3}}\text{,}\]
    which holds if and only if \(\alpha_3 \beta_2 \geq \alpha_2 \beta_3\). In sum, we have
    \[\frac{\beta_3}{\alpha_3} \leq \frac{\beta_2}{\alpha_2} \leq \frac{\beta_2 - \beta_{3}}{\alpha_2 - \alpha_{3}} < \frac{\beta_1}{\alpha_1} \leq \frac{\beta_1 - \beta_{2}}{\alpha_1 - \alpha_{2}}\text{.}\]

    Let \(x_1 = y_1 = 1\), \(x_3 = y_2 = 0\), and \(x_2 = y_3 = \eta > 0\). Then,
    \[\begin{split}
        \frac{\alpha_1 x_1 + \alpha_2 x_2 + \alpha_{3} x_{3}}{\alpha_1 y_1 + \alpha_2 y_2 + \alpha_{3} y_{3}} - \frac{\beta_1 x_1 + \beta_2 x_2 + \beta_3 x_3}{\beta_1 y_1 + \beta_2 y_2 + \beta_3 y_3} &> 0 \quad \Leftrightarrow\\
        \frac{\alpha_1 + \alpha_2 \eta}{\alpha_1 + \alpha_{3}\eta} - \frac{\beta_1 + \beta_2 \eta}{\beta_1 + \beta_3 \eta} &> 0 \quad \Leftrightarrow\\
        \eta\left(\beta_1\left(\alpha_2 - \alpha_3\right) -\alpha_1 \left(\beta_2 - \beta_3\right) + \eta \left(\alpha_2 \beta_3 - \alpha_3 \beta_2\right)\right) &> 0
    \end{split}\]
    If \(\alpha_2 \beta_3 = \alpha_3 \beta_2\), the last inequality holds for any strictly positive \(\eta\). Suppose instead \(\alpha_2 \beta_3 < \alpha_3 \beta_2\). Then, the last inequality holds for all 
    \[\eta \in \left(0,\frac{\beta_1\left(\alpha_2 - \alpha_3\right) -\alpha_1 \left(\beta_2 - \beta_3\right)}{\alpha_3 \beta_2 - \alpha_2 \beta_3}\right)\text{,}\]
    which interval is nonempty by assumption.

    \medskip

    \noindent \(\left(\Leftarrow\right)\) Now let \[\inf_{t \in \bar{\mathcal{T}}}\frac{\beta_t - \beta_{t+1}}{\alpha_t - \alpha_{t+1}} \geq \frac{\beta_1}{\alpha_1}\text{.}\]
    We want to show that for any positive prize sequences \(x\) and \(y\) with \(x \trianglerighteq y\) and \(\sum_{t=1}^{T} x_t = \sum_{t=1}^{T} y_t\), such that \(\sum_{t=1}^{T} \alpha_t y_t \neq 0\) and \(\sum_{t=1}^{T} \beta_t y_t \neq 0\),
    \[\frac{\sum_{t=1}^{T} \alpha_t x_t}{\sum_{t=1}^{T} \alpha_t y_t} \leq \frac{\sum_{t=1}^{T} \beta_t x_t}{\sum_{t=1}^{T} \beta_t y_t}\text{.}\]
    Recall also that, by assumption, \(\alpha_t, \beta_t > 0\) for \(t \in \mathcal{T}\).

    Lemma \ref{deteriorate} means that it suffices to show that for any \(x\) with \(\sum_{t=1}^{T} \alpha_t x_t \neq 0\) and \(\sum_{t=1}^{T} \beta_t x_t \neq 0\) and any binary deterioration of \(x\), \(z\), with \(\sum_{t=1}^{T} \alpha_t z_t \neq 0\) and \(\sum_{t=1}^{T} \beta_t z_t \neq 0\),
        \[\frac{\sum_{t=1}^{T} \beta_t z_t}{\sum_{t=1}^{T} \alpha_t z_t} \leq \frac{\sum_{t=1}^{T} \beta_t x_t}{\sum_{t=1}^{T} \alpha_t x_t}\text{.}\]
        
    Now take \[\frac{\sum_{t=1}^{T} \beta_t x_t}{\sum_{t=1}^{T} \alpha_t x_t}\text{,}\] with \(\sum_{t=1}^{T} \alpha_t x_t \neq 0\) and \(\sum_{t=1}^{T} \beta_t x_t \neq 0\) and an arbitrary binary deterioration of \(x\), \(z\), with \(\sum_{t=1}^{T} \alpha_t z_t \neq 0\) and \(\sum_{t=1}^{T} \beta_t z_t \neq 0\). That is, take some \(x_k\) and subtract \(\eta\), add \(\eta\) to some \(x_{s}\) with \(s, k \in \bar{\mathcal{T}}\) and \(s > k\). We have 
\[\frac{\sum_{t=1}^{T} \beta_t x_t}{\sum_{t=1}^{T} \alpha_t x_t} - \frac{\sum_{t=1}^{T} \beta_t z_t}{\sum_{t=1}^{T} \alpha_t z_t} = \frac{\sum_{t=1}^{T} \beta_t x_t}{\sum_{t=1}^{T} \alpha_t x_t} - \frac{\sum_{t=1}^{T} \beta_t x_t - \eta (\beta_k - \beta_s)}{\sum_{t=1}^{T} \alpha_t x_t - \eta (\alpha_k - \alpha_s)}\text{.}\]
The right-hand side of this is positive if and only if
\[\frac{\beta_k - \beta_s}{\alpha_k - \alpha_s} \geq \frac{\sum_{t=1}^{T} \beta_t x_t}{\sum_{t=1}^{T} \alpha_t x_t}\text{.}\]

Observe that, appealing to the mediant inequality, it suffices to show
\[\inf_{k,s \in \bar{\mathcal{T}}, s > k}\frac{\beta_k - \beta_s}{\alpha_k - \alpha_s} \geq \sup_{t \in \bar{\mathcal{T}}}\frac{\beta_t}{\alpha_t}\text{.}\]

\begin{claim}\label{claim3.9}
    For any \(k, s \in \bar{\mathcal{T}}\) with \(s > k\)
\[\frac{\beta_k - \beta_s}{\alpha_k - \alpha_s} \geq \min\left\{\frac{\beta_{k} - \beta_{k+1}}{\alpha_{k} - \alpha_{k+1}}, \frac{\beta_{k+1} - \beta_{k+2}}{\alpha_{k+1} - \alpha_{k+2}},\dots,\frac{\beta_{s-1} - \beta_{s}}{\alpha_{s-1} - \alpha_{s}}\right\}\text{.}\]
\end{claim}
\begin{proof}
As 
\[\beta_k - \beta_s = \left(\beta_k - \beta_{k+1}\right) + \left(\beta_{k+1} - \beta_{k+2}\right) + \cdots + \left(\beta_{s-1} - \beta_{s}\right)\text{,}\]
and
\[\alpha_k - \alpha_s = \left(\alpha_k - \alpha_{k+1}\right) + \left(\alpha_{k+1} - \alpha_{k+2}\right) + \cdots + \left(\alpha_{s-1} - \alpha_{s}\right)\text{,}\]
this follows from the mediant inequality.
\end{proof}

\begin{claim}\label{claim310}
    \[\inf_{t \in \bar{\mathcal{T}}} \frac{\beta_{t} - \beta_{t+1}}{\alpha_{t} - \alpha_{t+1}} \geq \frac{\beta_1}{\alpha_1} \quad \Rightarrow \quad \inf_{t \in \bar{\mathcal{T}}} \frac{\beta_{t} - \beta_{t+1}}{\alpha_{t} - \alpha_{t+1}} \geq \sup_{t \in \bar{\mathcal{T}}} \frac{\beta_t}{\alpha_t}\]
\end{claim}
\begin{proof}
    For any \(k \in \bar{\mathcal{T}}\),
    \[\inf_{t \in \bar{\mathcal{T}}} \frac{\beta_{t} - \beta_{t+1}}{\alpha_{t} - \alpha_{t+1}} \geq \frac{\beta_1}{\alpha_1} \ \Rightarrow \ \frac{\beta_{1} - \beta_{k}}{\alpha_{1} - \alpha_{k}} \geq \frac{\beta_1}{\alpha_1} \ \Rightarrow \ \frac{\beta_1}{\alpha_1} \geq \frac{\beta_k}{\alpha_k} \text{,}\]
    where the first implication uses Claim \ref{claim3.9}.
    \end{proof}
    From Claim \ref{claim3.9} we have
    \[\begin{split}
        \inf_{k,s \in \bar{\mathcal{T}}, s > k}\frac{\beta_k - \beta_s}{\alpha_k - \alpha_s} &\geq \inf_{k,s \in \bar{\mathcal{T}}, s > k}\min\left\{\frac{\beta_{k} - \beta_{k+1}}{\alpha_{k} - \alpha_{k+1}}, \frac{\beta_{k+1} - \beta_{k+2}}{\alpha_{k+1} - \alpha_{k+2}},\dots,\frac{\beta_{s-1} - \beta_{s}}{\alpha_{s-1} - \alpha_{s}}\right\}\\ &=\inf_{t \in \bar{\mathcal{T}}}\frac{\beta_t - \beta_{t+1}}{\alpha_t - \alpha_{t+1}}\text{.}
    \end{split}\]
    Finally, combining this with Claim \ref{claim310},
    we have
    \[\inf_{t \in \bar{\mathcal{T}}}\frac{\beta_t - \beta_{t+1}}{\alpha_t - \alpha_{t+1}} \geq \frac{\beta_1}{\alpha_1} \quad \Rightarrow \quad \inf_{k,s \in \bar{\mathcal{T}}, s > k}\frac{\beta_k - \beta_s}{\alpha_k - \alpha_s} \geq \sup_{t \in \bar{\mathcal{T}}}\frac{\beta_t}{\alpha_t}\text{,}\]
    and we conclude the theorem.\end{proof}

\subsection{Proposition \ref{transpatient} Proof}
\begin{proof}
        Let Michelle discount according to \(\gamma\), Alice and Bob according to \(\alpha\) and \(\beta\). Let \[\inf_{t \in \bar{\mathcal{T}}}\frac{\beta_t - \beta_{t+1}}{\alpha_t - \alpha_{t+1}} \geq \frac{\beta_1}{\alpha_1} \quad \text{and} \quad \inf_{t \in \bar{\mathcal{T}}}\frac{\gamma_t - \gamma_{t+1}}{\beta_t - \beta_{t+1}} \geq \frac{\gamma_1}{\beta_1}\text{.}\]
    Then, suppose for the sake of contradiction that
    \[\frac{\gamma_1}{\alpha_1} > \inf_{t \in \bar{\mathcal{T}}}\frac{\gamma_t - \gamma_{t+1}}{\alpha_t - \alpha_{t+1}}\text{.}\]
    Consequently
    \[\inf_{t \in \bar{\mathcal{T}}}\frac{\gamma_t - \gamma_{t+1}}{\beta_t - \beta_{t+1}} \geq \frac{\gamma_1}{\beta_1} > \frac{\alpha_1}{\beta_1}\inf_{t \in \bar{\mathcal{T}}}\frac{\gamma_t - \gamma_{t+1}}{\alpha_t - \alpha_{t+1}} \geq \frac{\inf_{t \in \bar{\mathcal{T}}}\frac{\gamma_t - \gamma_{t+1}}{\alpha_t - \alpha_{t+1}}}{\inf_{t \in \bar{\mathcal{T}}}\frac{\beta_t - \beta_{t+1}}{\alpha_t - \alpha_{t+1}}}\text{,}\]
    a contradiction.
    \end{proof}
    
\end{document}